\documentclass{article}
\usepackage{amsmath,amsfonts,amssymb,amsthm}
\usepackage{graphicx}
\usepackage{paralist}
\usepackage{booktabs}
\usepackage[round]{natbib}
\usepackage{multirow}
\usepackage{url}
\allowdisplaybreaks
\newcommand{\cl}{{\cal L}}

\newcommand{\cx}{{\cal X}}

\newcommand{\bsx}{\boldsymbol{x}}
\newcommand{\bsy}{\boldsymbol{y}}
\newcommand{\bsz}{\boldsymbol{z}}

\newcommand{\real}{\mathbb{R}}

\newcommand{\tran}{\mathsf{T}} % Transpose

\newcommand{\var}{{\mathrm{Var}}}
\newcommand{\cov}{{\mathrm{Cov}}}

\newcommand{\wh}{\widehat}
\newcommand{\wt}{\widetilde}

\renewcommand{\emptyset}{\varnothing}
\renewcommand{\ge}{\geqslant}
\renewcommand{\le}{\leqslant}

\newcommand{\dnorm}{\mathcal{N}}

\newcommand{\dustd}{\mathbf{U}} % for U(0,1) and U[0,1)

\newcommand{\e}{{\mathbb{E}}} % expectation
\newtheorem{theorem}{Theorem}

\newtheorem{lemma}{Lemma}
\newtheorem{corollary}{Corollary}

\theoremstyle{definition}

\newcommand{\ult}{\underline\tau}
\newcommand{\olt}{\bar\tau}
\newcommand{\simiid}{\stackrel{\mathrm{iid}}{\sim}}

\renewcommand{\cl}{\mathcal{L}}

\begin{document}

% Title portion
\title{Efficient estimation of the ANOVA mean dimension,
with an application to neural net classification}
\author{
Christopher R. Hoyt\\Stanford University
\and
Art B. Owen\\Stanford University
}
\date{December 2020}
\maketitle

\begin{abstract}
The mean dimension of a black box function of $d$
variables is a convenient
way to summarize the extent to which it is dominated by
high or low order interactions.  It is expressed in terms
of $2^d-1$ variance components but it can be written
as the sum of $d$ Sobol' indices that can be estimated
by leave one out methods.
We compare the variance of these leave one out methods:
a Gibbs sampler called winding stairs, a radial sampler
that changes each variable one at a time from a baseline,
and a naive sampler that never reuses function evaluations
and so costs about double the other methods.
For an additive function the radial and winding stairs methods are most efficient.
For a multiplicative function the naive method can easily
be most efficient if the factors have high kurtosis.
As an illustration we consider the mean dimension of
a neural network classifier of digits from the MNIST
data set. The classifier is a function of $784$ pixels.
For that problem, winding stairs is the best algorithm.
We find that inputs to the final softmax layer have mean dimensions
ranging from $1.35$ to $2.0$.
\end{abstract}


\section{Introduction}

The mean dimension of a square integrable
function quantifies
the extent to which higher order interactions among its
$d$ input variables are important.  At one
extreme, an additive function has mean dimension one
and this makes numerical tasks such
as optimization and integration much simpler.
It can also make it easier to compare the importance
of the inputs to a function and it simplifies some visualizations.
At the other extreme, a function that equals a $d$-fold
interaction has mean dimension $d$ and can be much
more difficult to study.

The mean dimension of a function can be expressed
as a certain sum of Sobol' indices which we introduce below.
There is an extensive literature on efficiently estimating
Sobol' indices
\citep{homm:salt:1996,
jans:1999,salt:2002,
mono:naud:mao:2006,
glen:isaa:2012,
jano:klei:lagn:node:prie:2014,
salt:anno:azzi:camp:ratt:tara:2010} and there
are additional references in \cite{puy:etal:2020}
which has a thorough empirical comparison of methods
for the total index, which is the one we use below.
In the case of mean dimension, the necessary
indices can be estimated numerically
by algorithms that change just one input variable at a time.
Two prominent strategies for this case are the winding stairs
estimator of \cite{jans:etal:1994} which runs a Gibbs
sampler over the input space and a radial strategy
of \cite{camp:salt:cri:2011}.

%The squared differences in a function's
%value arising from changing one input at a time can be used to estimate
%a certain Sobol' index described below. The mean dimension is
%defined in terms of a sum of such Sobol' indices.
When estimating the mean dimension, a special consideration
arises.  Since it is a sum of $d$ Sobol' indices,
there are $O(d^2)$ covariances to consider
and they can greatly affect the
efficiency of the estimation strategy.  Sometimes a naive approach
that uses roughly twice as many function evaluations can
be more efficient than winding stairs because it eliminates
all of those covariances.

The outline of this paper is as follows.
Section~\ref{sec:notation} introduces some notation, and defines
the ANOVA decomposition, Sobol' indices and the mean dimension.
Section~\ref{sec:strategies} presents three strategies for sampling
pairs of input points that differ in just one component.  A naive
method takes $2Nd$ function evaluations to get $N$ such pairs
of points for each of $d$ input variables. It never reuses any
function values. A radial strategy
\citep{camp:salt:cri:2011}
uses $N(d+1)$ function evaluations
in which $N$ baseline points each get paired with $d$ other points
that change one of the inputs.  The third strategy is winding stairs
\citep{jans:etal:1994} which uses $Nd+1$ function evaluations.
Section~\ref{sec:addmult} compares
the variances of mean dimension estimates based on these strategies.
Those variances involve fourth moments of the original function.
We consider additive and multiplicative functions. % and the Euclidean norm
For additive functions all three methods have the same variance
making the naive method inefficient by a factor of about $2$
for large $d$.
For some functions, methods that save function
evaluations by reusing some of them can introduce positive
correlations yielding a less efficient estimate.  We find that the presence
of factors with high kurtoses can decrease the value of reusing evaluations.
%Section~\ref{sec:quadratic} considers quadratic forms in Gaussian variables.
Section~\ref{sec:neural} presents an example where we measure
the mean dimension of a neural network classifier designed
to predict a digit $0$ through $9$ based on $784$ pixels.
It was interesting to see the mean dimensions fall in the range
from $1.35$ to $2.0$ for the penultimate layer of the network,
suggesting that the information from those pixels is being
used mostly one or two or three at a time.
For instance, there cannot be any meaningfully
large interactions of $100$ or more inputs.
Section~\ref{sec:discussion} makes some concluding
remarks.  Notably, the circumstances that make the radial method
inferior to the naive method or winding stairs for computing mean dimension
serve to make it superior to them
for some other uncertainty quantification tasks.
We also discuss randomized quasi-Monte Carlo sampling alternatives
and make brief comments about dependent inputs.
Finally, there is an Appendix in which
we provide a more detailed analysis of winding stairs.

\section{Notation}\label{sec:notation}

We begin with the analysis of variance (ANOVA) decomposition
for a function $f:\cx\to\real$ where
$\cx = \prod_{j=1}^d\cx_j$.
We let $\bsx = (x_1,\dots,x_d)$ where $x_j\in\cx_j$.
The ANOVA is defined in terms of a distribution on $\cx$
for which the $x_j$ are independent
and for which $\e( f(\bsx)^2)<\infty$.
The $\cx_j$ are ordinarily subsets of $\real$ but the ANOVA
is well defined for more general domains.
We let $P$ denote the distribution of $\bsx$
and $P_j$ denote the distribution of $x_j$.
The ANOVA of $[0,1]^d$ was proposed
by \cite{hoef:1948} for $U$-statistics,
and by \cite{sobo:1969} for numerical integration.
It is well known in statistics
following \cite{efro:stei:1981} where the ANOVA
underlies the Efron-Stein inequality for the jackknife.

We will use $1{:}d$ as a short form for $\{1,2,\dots,d\}$.
For sets $u\subseteq1{:}d$, their cardinality is $|u|$
and their complement $1{:}d\setminus u$ is denoted by $-u$.
The components $x_j$ for $j\in u$ are collectively denoted by $\bsx_u$.
We will use hybrid points that merge components from two
other points.  The point $\bsy=\bsx_u{:}\bsz_{-u}$
has $y_j = x_j$ for $j\in u$ and $y_j = z_j$ for $j\not \in u$.
It is typographically convenient to replace singletons $\{j\}$
by $j$, especially within subscripts.

The ANOVA decomposition writes
$f(\bsx) = \sum_{u\subseteq1{:}d}f_u(\bsx)$
where the `effect' $f_u$ depends on $\bsx$ only through $\bsx_u$.
The first term is $f_\emptyset(\bsx) = \e(f(\bsx))$ and
the others are defined recursively via
$$
f_u(\bsx) = \e\Bigl( f(\bsx) -\sum_{v\subsetneq u}f_v(\bsx)\!\bigm|\! \bsx_u\Bigr).
$$
The variance component for $u$ is
$$
\sigma^2_u \equiv \var(f_u(\bsx)) =
\begin{cases}
\e( f_u(\bsx)^2), & u\ne\emptyset\\
0, & u=\emptyset.
\end{cases}
$$
The effects are orthogonal under $P$
and $\sigma^2=\var(f(\bsx))=\sum_u\sigma^2_u$.
We will assume that $\sigma^2>0$ in order to make some quantities
well defined.

Sobol' indices quantify importance of
subsets of input variables on $f$.
They are a primary method in global sensitivity analysis \citep
{salt:ratt:andr:camp:cari:gate:sais:tara:2008,
ioos:lema:2015, borg:plis:2016}.
Lower and upper Sobol' indices are
$$
\ult^2_u = \sum_{v\subseteq u}\sigma^2_v
\quad\text{and}\quad
\olt^2_u = \sum_{v\cap u\ne\emptyset}\sigma^2_v,
$$
respectively.
The lower index is from \cite{sobo:1990,sobo:1993},
while the upper index was first used by \cite{homm:salt:1996}.
Both indices are commonly normalized,
with $\ult^2_u/\sigma^2$ known as the closed
index and $\olt^2_u/\sigma^2$ is called the total index.
Normalized indices are between $0$ and $1$ giving
them interpretations as a proportion of variance explained,
similar to $R^2$ from regression models.
The Sobol' indices $\ult^2_{j}$ and $\olt^2_{j}$
for singletons $\{j\}$ are of special interest.

Sobol' indices satisfy these identities
\begin{align*}
\ult^2_u & =\e\bigl( f(\bsx)f(\bsx_u{:}\bsz_{-u})\bigr)-\mu^2\\
& =\e\bigl( f(\bsx)(f(\bsx_u{:}\bsz_{-u})-f(\bsz))\bigr)\quad\text{and}\\
\olt^2_u &= \frac12\e\bigl( (f(\bsx)-f(\bsx_{-u}{:}\bsz_u))^2\bigr),
\end{align*}
when $\bsz$ is an independent copy of $\bsx$.
Those identities make it possible to estimate $\ult^2_u$ and $\olt^2_u$ by
Monte Carlo or quasi-Monte
Carlo sampling without explicitly computing estimates of any
of the effects $f_v$.
The first identity is due to \cite{sobo:1993}.
The second was proposed independently by \cite{salt:2002} and \cite{maun:2002}.
The third identity underlies an estimator of the total
index from \cite{jans:1999}.
The numerator in the estimate of $\olt^2_j$ from
\cite{homm:salt:1996} is based on the identity
of Sobol' along with $\olt^2_j = \sigma^2-\ult^2_{-j}$.

The mean dimension of $f$ is
$$
\nu(f) = \sum_{u\subseteq1{:}d}\frac{|u|\sigma^2_u}{\sigma^2}.
$$
It satisfies $1\le \nu(f)\le d$.
A low mean dimension indicates that $f$ is dominated by low
order ANOVA terms, a favorable property for some numerical problems.

An easy identity from \cite{meandim} shows that
$\sum_{u\subseteq1{:}d}|u|\sigma^2_u = \sum_{j=1}^d\olt^2_j$.
Then the mean dimension of $f$ is
$$
\nu(f)\equiv\frac1{\sigma^2}\sum_{j=1}^d \olt^2_j,\quad
\text{for}\quad
\olt^2_j =
\frac12 \e\bigl( (f(\bsx) - f(\bsx_{-j}{:}\bsz_j)\bigr)^2.
$$
Although the mean dimension combines $2^d-1$
nonzero variances it can be computed from $d$
Sobol' indices (and the total variance $\sigma^2$).

We can get a Monte Carlo estimate of
the numerator of $\nu(f)$
by summing estimates of $\olt^2_j$
such as
\begin{align}\label{eq:olthat}
\frac1{2N}\sum_{i=1}^N\bigl( f(\bsx_i)-f(\bsx_{i,-j}{:}\bsz_{i,j})\bigr)^2
\end{align}
for independent random points $\bsx_i,\bsz_i\sim P$.
Equation~\eqref{eq:olthat} corresponds to the strategy
that \cite{jans:1999} uses to estimate the numerator of
the normalized total sensitivity
index for $x_j$. The Jansen estimator was one of the
best performers in \cite{puy:etal:2020}.

There is more than one way to arrange the computation
that sums~\eqref{eq:olthat} over $j=1,\dots,d$.
When computing a list of Sobol' indices it is advantageous
to reuse many of the function values. See \cite{salt:2002}
for some strategies.
When there are $O(d)$ total indices to sum,
there are $O(d^2)$ covariances to consider and
that is the issue we focus on here.

\section{Estimation strategies}\label{sec:strategies}

Equation~\eqref{eq:olthat} gives an estimate of
$\olt^2_j$ evaluating $f$ at pairs of points
that differ only in their $j$'th coordinate.
An estimate for the numerator of $\nu(f)$ sums these estimates.
We have found empirically and somewhat surprisingly
that different sample methods for computing
the numerator $\sum_j\olt^2_j$ can have markedly different variances,
even when they are all of \cite{jans:1999} type.

A naive implementation uses $2Nd$
function evaluations
taking $\bsx_i,\bsz_i$ independent for $i=1,\dots,N$
for each of $j=1,\dots,d$.  In that strategy, the point $\bsx_i$
in~\eqref{eq:olthat} is actually different for each $j$.
Such a naive implementation is wasteful.
We could instead use the same $\bsx_i$ and $\bsz_i$ for all $j=1,\dots,d$
in the radial method of
\cite{camp:salt:cri:2011}.
This takes $N(d+1)$ evaluations of $f$.
A third strategy is known as `winding stairs'
\citep{jans:etal:1994}.
%\citep{chan:salt:tara:2000}.
The data come from a Gibbs sampler, that in
its most basic form
changes inputs to $f$ one at a time
changing indices in this order:
 $j=1,\dots,d,1,\dots,d,\cdots,1,\dots,d$.
It uses only $Nd+1$ evaluations of $f$.
These three approaches are illustrated in Figure~\ref{fig:sampler}.
We will also consider a variant of winding stairs that randomly
refreshes after every block of $d+1$ evaluations.

\begin{figure}
\centering
\includegraphics[width=.9\hsize]{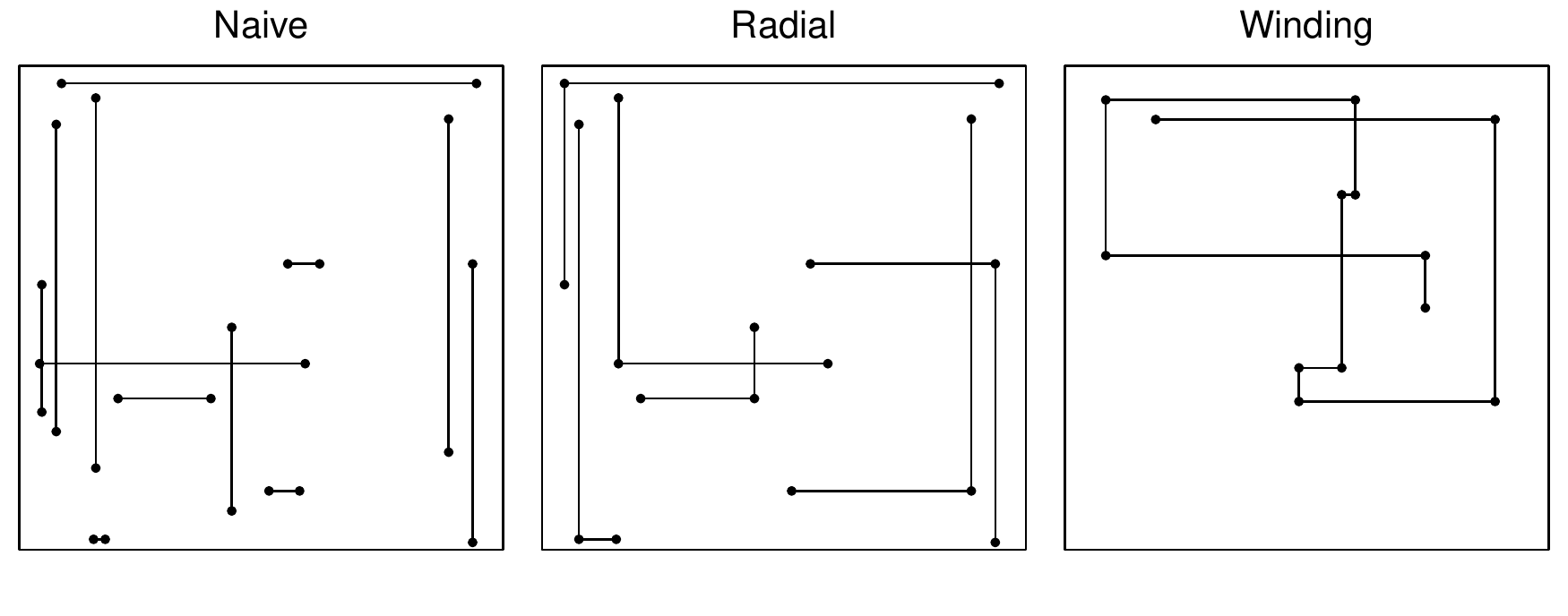}
\caption{\label{fig:sampler}
Examples of three input sets
to compute $\delta =\sum_{j=1}^d\olt^2_j$ when $d=2$.
The naive estimate uses $dN$ pairs of points,
$N$ pairs for each of $d$ variables.
Each edge connects a pair of points used in the estimate.
The radial estimate uses $N$ baseline points
and $d$ comparison points for each of them.
The winding stairs estimates sequentially
changes one input at a time.
}
\end{figure}

First we compare the naive to the radial strategy.
For $\nu = \sum_j\olt^2_j/\sigma^2$ we concentrate
on estimation strategies for the numerator
$$
\delta = \sigma^2\nu = \sum_{j=1}^d\olt^2_j.
$$
This quantity is much more challenging to estimate
than the denominator $\sigma^2$, especially for
large $d$,  as it involves $d^2$ covariances.

The naive sampler takes
\begin{align}\label{eq:naivedelta}
\hat\delta
&=\sum_{j=1}^d \wh\olt^2_j\quad\text{where}\quad
\wh\olt^2_j=
\frac1{2N}\sum_{i=1}^N
\bigl(f(\bsx^{(j)}_i)-f(\bsx^{(j)}_{i,-j}{:}\bsz_{i,j})\bigr)^2
\end{align}
with independent $\bsz_i,\bsx^{(j)}_i\sim P$
for $i=1,\dots,N$ and $j=1,\dots,d$.
It takes $N(d+1)$ input vectors and
$2Nd$ evaluations of $f$.

The radial sampler takes
\begin{align}\label{eq:radialdelta}
\tilde\delta
&=\sum_{j=1}^d \wt\olt^2_j\quad\text{where}\quad
\wt\olt^2_j=
\frac1{2N}\sum_{i=1}^N
\bigl(f(\bsx_i)-f(\bsx_{i,-j}{:}\bsz_{i,j})\bigr)^2,
\end{align}
for independent $\bsx_i,\bsz_i\sim P$,
$i=1,\dots,N$.

For $f\in L^2(P)$ both $\tilde\delta$
and $\hat\delta$ converge to $\delta=\nu\sigma^2$
as $N\to\infty$ by the law of large numbers.
To compare accuracy of these estimates we
assume also that $f\in L^4(P)$. Then $\e( f(\bsx)^4)<\infty$
and both estimates have variances that are $O(1/N)$.

A first comparison is that
\begin{align}\label{eq:nocov}
\begin{split}
\var( \tilde\delta)
& = \sum_{j=1}^d \var(\wt\olt^2_j)
+2\sum_{1\le j<k\le d}
\cov(\wt\olt^2_j, \wt\olt^2_k),\quad\text{while}\\
\var( \hat\delta)
& = \sum_{j=1}^d \var(\wh\olt^2_j)
+2\sum_{1\le j<k\le d}
\cov(\wh\olt^2_j, \wh\olt^2_k)\\
& = \sum_{j=1}^d \var(\wh\olt^2_j)
\end{split}
\end{align}
by independence of $(\bsx^{(j)}_i,\bsz_{i,j})$
from  $(\bsx^{(k)}_i,\bsz_{i,k})$ for $j\ne k$.
What we see from~\eqref{eq:nocov} is that while
the naive estimate uses about twice as many
function evaluations, the radial estimate sums
$d$ times as many terms.
The off diagonal covariances do not have to be
very large for us to have
$\var(\tilde\delta) >2\var(\hat\delta)$,
in which case $\hat\delta$ becomes the
more efficient estimate despite using more function evaluations.
Intuitively, each time $f(\bsx_i)$ takes an unusually large or small
value it could make a large contribution to all $d$ of $\wt{\olt}^2_j$
and this can result in $O(d^2)$ positive covariances.  We
study this effect more precisely below giving additional assumptions
under which $\cov(\wt{\olt}_j^2, \wt{\olt}_k^2)>0$.
We also have a numerical counter-example
at the end of this section, and
so this positive covariance does not hold for all $f\in L^4(P)$.

%\section{Winding stairs}\label{sec:winding}

The winding stairs algorithm
starts at $\bsx_0\sim P$ and then
makes a sequence of single variable
changes to generate $\bsx_i$ for $i>0$.
We let $\ell(i)\in 1{:}d$ be the index of
the component
that is changed at step $i$. The new values
are independent samples $z_i \sim P_{\ell(i)}$.
That is, for $i>0$
$$
\bsx_{i,j} =\begin{cases} z_i,
&j =\ell(i)\\
\bsx_{i-1,j}, &j\ne \ell(i).
\end{cases}
$$
We have a special interest in the case where $P=\dnorm(0,I)$,
for which each $P_j$ is $\dnorm(0,1)$.

The indices $\ell(i)$ can be either deterministic
or random.  We let $\cl$ be the entire
collection of $\ell(i)$.  We assume that
the entire collection of $z_i$ are independent
of~$\cl$.
The most simple deterministic update has
$\ell(i) = 1+(i-1\mod d)$ and it cycles through
all indices $j\in1{:}d$ in order.
The simplest random update has
$\ell(i)\simiid \dustd(1{:}d)$.
In usual Gibbs sampling it would be
better to take
$\ell(i)\simiid \dustd(1{:}d\setminus\{\ell(i-1)\})$
for $i\ge2$.  Here because we are accumulating
squared differences it is not very harmful to have
$\ell(i)=\ell(i-1)$.
The vector $\bsx_i$ contains $d$ independently
sampled Gaussian random variables. Which ones
those are, depends on $\cl$.  Because $\bsx\sim\dnorm(0,I)$
conditionally on $\cl$ it also has that distribution
unconditionally.

Letting $e_j$ be the $j$'th unit vector in $\real^d$ we
can write
$$
\bsx_i = \bsx_{i-1} + (z_i-x_{i-1,\ell(i)})e_{\ell(i)}.
$$
If $\ell(i)\sim\dustd(1{:}d)$, then the distribution
of  $\bsx_i$ given $\bsx_{i-1}$ is a mixture of $d$
different Gaussian distributions, one for each
value of $\ell(i)$. As a result $\bsy_i =(\bsx_i^\tran,\bsx_{i-1}^\tran)^\tran$
does not then have a multivariate Gaussian distribution
and is harder to study.
For this reason, we focus on the deterministic update.

In the deterministic update we find that
any finite set of $\bsx_i$ or $\bsy_i$ has a multivariate
Gaussian distribution.
We also know that $\bsx_i$ and $\bsx_{i+k}$ are
independent for $k\ge d$ because after $k$ steps
all components of $\bsx_i$ have been replaced
by new $z_i$ values.
It remains to consider the correlations among a block
of $d+1$ consecutive vectors.
Those depend on the pattern of shared components
within different observations as illustrated in the following diagram:
\begin{align}\label{eq:windingdiagram}
\begin{matrix}
\bsx_d & \bsx_{d+1} & \bsx_{d+2} & \cdots & \bsx_{2d-1} &\bsx_{2d}\\[.25ex]
\| &\|&\|& &\|&\|
\\[.6ex]
\begin{pmatrix}
z_1\\
z_2\\
z_3\\
\vdots\\
z_{d-1}\\
z_d
\end{pmatrix}
&
\begin{pmatrix}
z_{d+1}\\
z_{2}\\
z_3\\
\vdots\\
z_{d-1}\\
z_d
\end{pmatrix}
&
\begin{pmatrix}
z_{d+1}\\
z_{d+2}\\
z_3\\
\vdots\\
z_{d-1}\\
z_d
\end{pmatrix}
&\cdots&
\begin{pmatrix}
z_{d+1}\\
z_{d+2}\\
z_{d+3}\\
\vdots\\
z_{2d-1}\\
z_d
\end{pmatrix}
&
\begin{pmatrix}
z_{d+1}\\
z_{d+2}\\
\vdots\\
z_{2d-1}\\
z_{2d}
\end{pmatrix}
\end{matrix}.
\end{align}

For $i\ge d$ and $j=1,\dots,d$ we can write
\begin{align}\label{eq:rofij}
\bsx_{i,j} = z_{r(i,j)}
\quad\text{where}\quad r(i,j) = d\Big\lfloor\frac{i-j}d\Bigr\rfloor+j.
\end{align}
It is convenient to use~\eqref{eq:rofij} for
all $i\ge0$ which is equivalent
to initializing the sampler at
$\bsx_0 = (z_{-(d-1)},z_{-(d-2)},\dots,z_{-1},z_0)^\tran$.
Equation~\eqref{eq:rofij}
holds for any independent $z_i \sim P_{\ell(i)}$
and does not depend
on our choice of $P_j=\dnorm(0,1)$.

The winding stairs estimate of $\delta$ is
\begin{align}\label{eq:windingdelta}
\check\delta & = \sum_{j=1}^d\check{\olt}^2_j\quad\text{for}\quad
\check{\olt}^2_j  = \frac1{2N}\sum_{i=1}^N\Delta_{d(i-1)+j}^2,
\end{align}
where $\Delta_r = f(\bsx_r)-f(\bsx_{r-1})$.
We will see that the covariances of $\check{\olt}^2_j$
and $\check{\olt}^2_k$ depend on the pattern of common
components among the $\bsx_i$. In our special
case functions certain kurtoses have an impact on the
variance of winding stairs estimates.

A useful variant of winding stairs simply makes
$N$ independent replicates of the $d+1$ vectors
shown in~\eqref{eq:windingdiagram}.
That raises the number of function
evaluations from $Nd+1$
to $N(d+1)$.  It uses $N$
independent Markov chains of length $d+1$.
%and their independence makes it straightforward to compute sample variances.
For large $d$ the increased computation is negligible.
In original winding stairs, each
squared difference $\Delta_i^2=(f(\bsx_i)-f(\bsx_{i-1}))^2$
can be correlated with up to $2(d-1)$ other
squared differences.  In truncated
winding stairs, it can only
be correlated with $d-1$ other squared differences.
We denote the resulting estimate by
$\ddot\delta$ which is a sum of
${\ddot\olt}^2_j$.

For $d=2$ this truncated winding stairs method
is the same as radial sampling.
For $d\ge3$ they are different. For instance the
value of $f$ at the radial
point is compared to $f$ at $d$ other points in
the radial method while no function value is compared
to more than $2$ others in the variant of winding stairs.
See Figure~\ref{fig:3d} for an illustration when $d=3$.
\begin{figure}\centering
\includegraphics[width=.9\hsize]{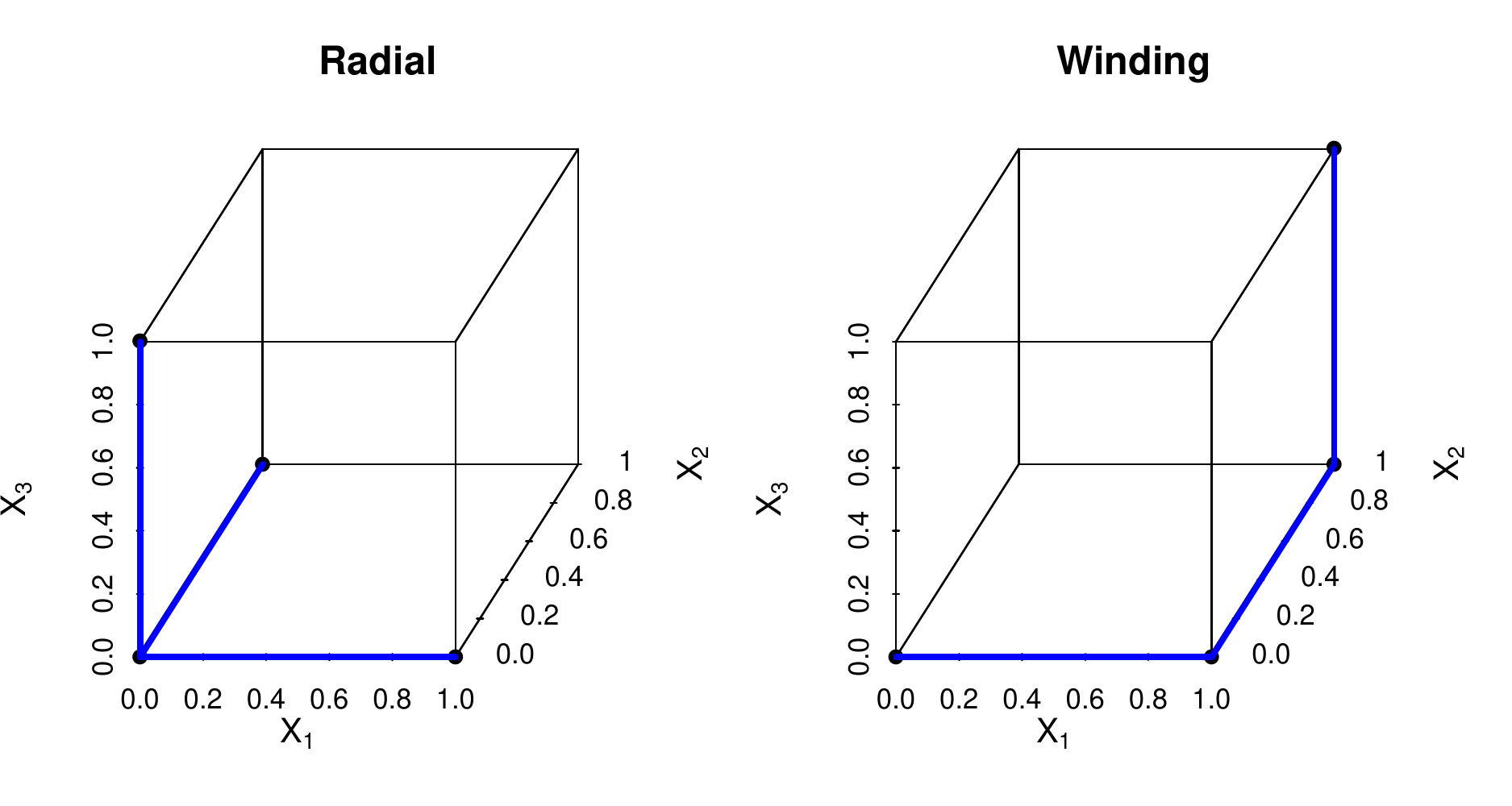}
\caption{\label{fig:3d}
The left figure shows example input points used by
the radial method for $d=3$, with thick edges connecting input
points used to form differences in $f$.  The right figure shows
the same for the truncated variant of winding stairs.
}
\end{figure}

% For instance the radial method might use
% $(f(\bsx_0)-f(\bsx_1))^2$,
% $(f(\bsx_0)-f(\bsx_2))^2$ and
% $(f(\bsx_0)-f(\bsx_3))^2$
% in $\olt^2_1$, $\olt^2_2$ and $\olt^2_3$
% where each $\bsx_j$ differs from $\bsx_0$ only
% in component $j$.
% The truncated winding stairs method
% would use $(f(\bsx_0)-f(\bsx_1))^2$,
% $(f(\bsx_1)-f(\bsx_2))^2$ and
% $(f(\bsx_2)-f(\bsx_3))^2$,
% where the $j$'th pair compares inputs
% that differ only in their $j$'th components.
% With the radial method, all $d$ squared differences
% involve the common radial point and that
% influences the correlations among those squared
% differences.  With truncated winding stairs, no input
% point is involved in more than two of the squared differences.

In section~\ref{sec:multfun} we present some multiplicative
functions where the naive estimator of $\delta$ has much
less than half of the variance of the radial estimator.  To complete this section
we exhibit a numerical example where the naive estimator has increased
variance which must mean that the correlations induced by the
radial and winding estimators are at least slightly negative.
The integrand is simply $f(\bsx)=\Vert\bsx\Vert_2$ for
$\bsx\sim\dnorm(0,I)$ in $d$ dimensions.
Figure~\ref{fig:twonorm} shows results. We used $N=10^6$
evaluations to show that (truncated) winding stairs and radial
sampling both have smaller variance than the naive algorithm
for estimating $\delta$.  We also see extremely small mean
dimensions for $f(\bsx)$ that decrease as $d$ increases.
It relates to some work in progress studying mean dimension
of radial basis functions as a counterpart to \cite{hoyt:owen:2020}
on mean dimension of ridge functions.
The visible noise in that figure stems from the mean
dimensions all being so very close to $1$ that the vertical range
is quite small. The estimate for $d=1$
is roughly $0.9983$ where the true value must be $1$.

\begin{figure}
\centering
\includegraphics[width=.9\hsize]{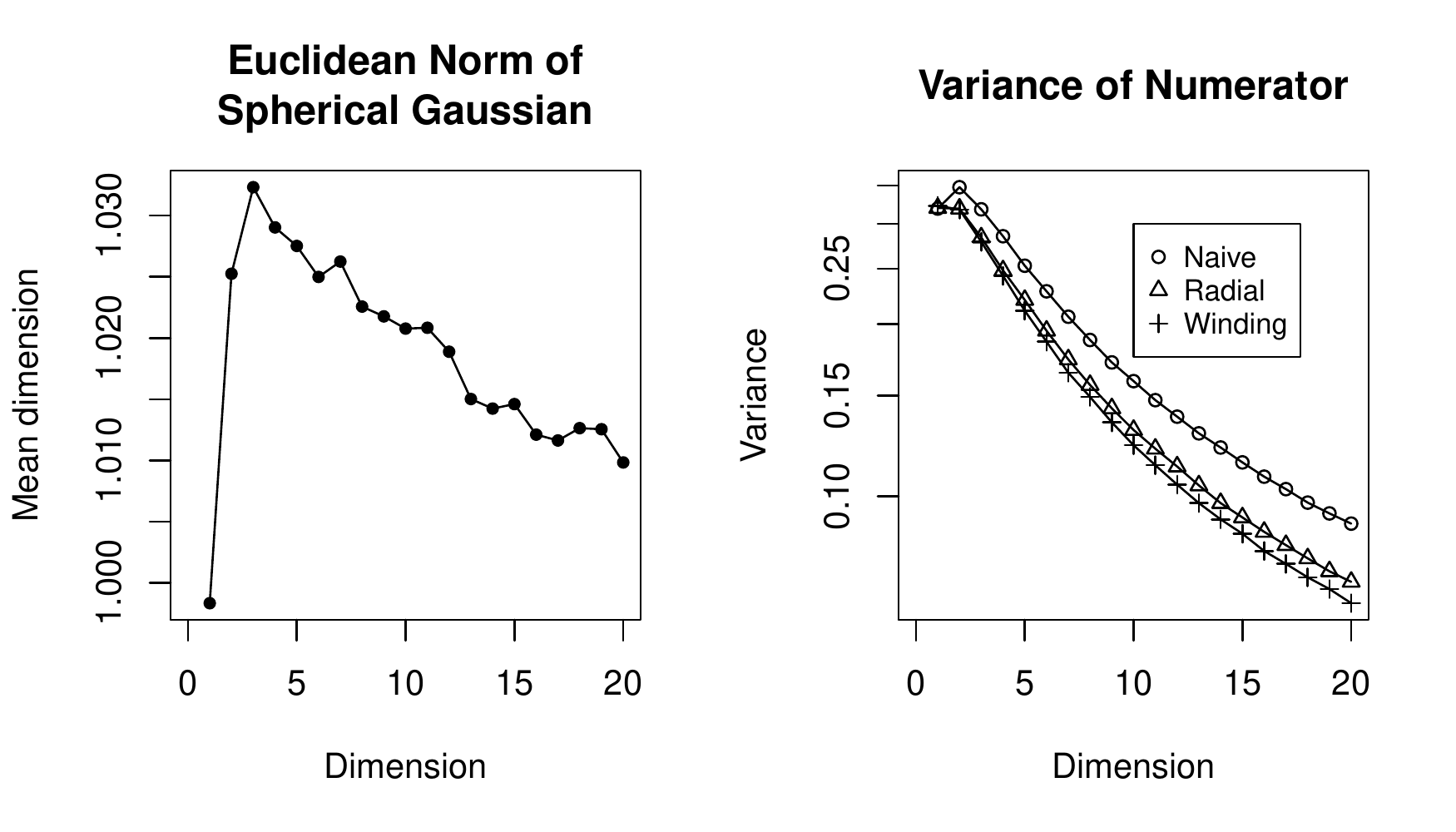}
\caption{\label{fig:twonorm}
The left panel shows low and mostly decreasing
estimates of $\nu(f)$ versus
dimension for $f(\bsx)=\Vert\bsx\Vert_2$ when $\bsx\sim\dnorm(0,I)$.
The right panel shows variances of estimates of $\delta$ for this function.
}
\end{figure}

\section{Additive and multiplicative functions}\label{sec:addmult}
The variances of quadratic functions
of the $f(\bsx_i)$ values such as $\hat\delta$, $\tilde\delta$
and $\check\delta$, involve
fourth moments of the original function. Whereas $2^d$
variance components are sufficient to define Sobol'
indices and numerous generalizations, fourth moments
do not simplify nearly as much from orthogonality
and involve considerably more quantities.
While distinct pairs of ANOVA effects are orthogonal,
we find for non-empty $u,v,w\subset1{:}d$ that
$$
%\e\Biggl( \prod_{\ell=1}^4f_{u_\ell}(\bsx)\Biggr)
\e\bigl( f_u(\bsx)f_v(\bsx)f_w(\bsx)\bigr)
$$
does not in general vanish when
$u\subset v\cup w$,
$v\subset u\cup w$
and $w\subset u\cup v$ all hold.
This `chaining phenomenon' is worse
for products of four effects: the number
of non-vanishing combinations rises
even more quickly with $d$.
The chaining problem also comes up if
we expand $f$ in an orthonormal basis for $L^2(P)$
and then look at fourth moments.

In this section we investigate some special functional
forms. The first is an additive model
\begin{align}\label{eq:fadd}
f_A(\bsx) = \mu+\sum_{j=1}^dg_j(x_j)
\end{align}
where $\e(g_j(x_j))=0$.
An additive model with finite variance has mean
dimension $\nu(f_A)=1$.  It represents one extreme in
terms of mean dimension.
The second function we consider is a product model
\begin{align}\label{eq:fprod}
f_P(\bsx) = \prod_{j=1}^dg_j(x_j)
\end{align}
where $\e(g_j(x_j))=\mu_j$ and
$\var(g_j(x_j))=\sigma^2_j$.
Product functions are frequently used as test functions.
For instance, Sobol's $g$-function
\citep{salt:sobo:1995} is the product
$\prod_{j=1}^d(|4x_j-2|+a_j)/(1+a_j)$
in which later authors make various choices for the constants $a_j$.

If all $\mu_j=0$ then $\nu(f_P)=d$.
In general, the mean dimension of a product function is
$$
\nu(f_P) = \frac{\sum_{j=1}^d
\sigma^2_j/(\mu_j^2+\sigma^2_j)}
{1-\prod_{j=1}^d\mu^2_j/(\mu^2_j+\sigma^2_j)}.
$$
See~\cite{dimdist}.

Additive and multiplicative functions comprise
two extremes in mean dimension.  Additive
functions always have mean dimension $1$.
While multiplicative functions can have any
mean dimension in the interval $(1,d]$ they
are easily engineered to provide functions
with mean dimension $d$ by setting all $\mu_j=0$.

\subsection{Additive functions}

We will use
Lemma~\ref{lem:mu4s} below  to compare the variances
of our mean dimension estimators for additive functions.
For these, we  need
the kurtosis of some random variables.
Recall that the kurtosis of a random variable
$Y$ with variance $\sigma^2>0$ is $\kappa = \e(( Y-\mu)^4)/\sigma^4-3$
which can be infinite.  As points of reference, if $Y$ is Gaussian then $\kappa=0$
and if $Y$ has a uniform distribution then $\kappa=-6/5$ and the
smallest possible kurtosis is $-2$.
%For a random variable $Y$, define
%$\mu_{2y} = \e(Y^2)$, $\mu_{3y}=\e(Y^3)$, $\mu_{4y} =\e(Y^4)$,
%the skewness $\gamma = \e((Y-\mu)^3)/\sigma^3$
%and the kurtosis
%$\kappa=\e( (Y-\mu)^4)/\sigma^4-3$.
%Gaussian random variables have $\gamma=\kappa=0$.
%%%%%We write $\mu_{jy}$ for $j=2,3,4$ instead of $\mu_j$
%because the latter is sometimes used as the mean of
%a function $g_j$.

\begin{lemma}\label{lem:mu4s}
Let $Y_1,Y_2,Y_3,Y_4$ be independent identically
distributed random variables with variance $\sigma^2$ and kurtosis $\kappa$.
Then
\begin{align*}
\e\bigl( (Y_1-Y_2)^4) & =(12+2\kappa)\sigma^4\\
\var((Y_1-Y_2)^2) & =
(8+2\kappa)\sigma^4\\
\e\bigl((Y_1-Y_2)^2(Y_3-Y_4)^2\bigr)
&= 4\sigma^4\\
\e\bigl((Y_1-Y_2)^2(Y_1-Y_3)^2\bigr)
&=(6+\kappa)\sigma^4.
\end{align*}
\end{lemma}
\begin{proof}
These follow directly from independence of
the $Y_j$ and the definitions of variance and kurtosis.
\end{proof}

\begin{theorem}\label{thm:additivecase}
For the additive function $f_A$ of~\eqref{eq:fadd},
\begin{align}
\var( \tilde\delta)
=\var( \hat\delta)=\var( \ddot\delta) &=
\frac1N\sum_{j=1}^d\Bigl(2+\frac{\kappa_j}2\Bigr)\sigma^4_j\label{eq:varnaivnum}
\end{align}
and
\begin{align}
\var( \check\delta)
&=\var(\ddot\delta)+\frac{N-1}{2N^2}\sum_{j=1}^d(\kappa_j+2)\sigma^4_j.
\end{align}
\end{theorem}

\begin{proof}
The winding stairs results for $\check\delta$ and $\ddot\delta$
quoted above are proved in Theorem~\ref{thm:additivecasewinding}
of the Appendix.
For the naive estimate,
$\wh\olt^2_j$ is independent of $\wh\olt^2_k$
when $j\ne k$ as remarked upon
at~\eqref{eq:nocov}.
For an additive function
$$
f_A( \bsx_i)
-f_A( \bsx_{i,-j}{:}\bsz_{i,j})
= g_j(x_{ij})-g_j(z_{ij})
$$
is independent of
$g_k(x_{ik})-g_k(z_{ik})$ for $j\ne k$
and so the radial estimate has the
same independence property as the
naive estimate.
Therefore
\begin{align*}
\var( \wh\olt^2_j)=\var( \wt\olt^2_j)
&=\frac1{4N}
\var\bigl( (g_j(x_{1j})-g_j(z_{1j}))^2\bigr)
\end{align*}
and using Lemma~\ref{lem:mu4s},
$\var( (g_j(x_{1j})-g_j(z_{1j}))^2)
=(8+2\kappa_j)\sigma^4_j$.
\end{proof}

If $f(\bsx)$ is additive, then Theorem~\ref{thm:additivecase}
shows that the radial method is better than the naive one.
They have the same variance but the naive method uses
roughly twice as many function evaluations.  If the function is
nearly additive, then it is reasonable to expect the variances
to be nearly equal and the radial method to be superior.
Because $\kappa_j\ge-2$ always holds, the theorem
shows an advantage to truncated winding stairs over plain winding stairs.

\subsection{Multiplicative functions}\label{sec:multfun}
We turn next to functions of product form.
For the factors $g_j(x_j)$ defining $f$ in equation~\eqref{eq:fprod},
we let $\mu_{2j}=\e(g_j(x_j)^2)$,
$\mu_{3j}=\e(g_j(x_j)^3)$ and
$\mu_{4j}=\e(g_j(x_j)^4)$.
To simplify some expressions for winding stairs
we adopt the conventions that for $1\le j<k\le d$
and quantities $q_\ell$,
$\prod_{\ell\in(j,k)}q_\ell$
means $\prod_{\ell=j+1}^{k-1}q_\ell$ and $\prod_{\ell\not\in[j,k]}q_\ell$
means $\prod_{\ell=1}^{j-1}q_\ell\times\prod_{\ell=k+1}^dq_\ell$,
with products over empty index sets equal to one.

\begin{theorem}\label{thm:multiplicativecase}
For the product function $f_P$ of~\eqref{eq:fprod},
\begin{align}
\var( \hat\delta) &=
\frac1N\sum_{j=1}^d\sigma_j^4\Bigl(\Bigl(3+\frac{\kappa_j}2 \Bigr)\prod_{\ell\ne j}\mu_{4\ell}
-\prod_{\ell\ne j}\mu_{2\ell}^2\Bigr) \quad\text{and}\label{eq:varnaivprod}\\
\var( \tilde\delta)
&= \var(\hat\delta)
+\frac{2}{N}\sum_{j<k}
\Bigl(\frac{\eta_j\eta_k}4-\sigma^2_j\sigma^2_k\mu_{2j}\mu_{2k}\Bigr)
\prod_{\ell\not\in\{j,k\}}\mu_{4\ell},
\label{eq:varradiprod}
\end{align}
where
$\eta_j=\e( g_j(x_j)^2(g_j(x_j)-g_j(z_j))^2)
=\mu_{4j}-2\mu_j\mu_{3j}+\mu_{2j}^2$,
for independent $x_j,z_j\sim P_j$.
The winding stairs estimates satisfy
\begin{align}
\label{eq:varmultdiswindnummain}
\var( \ddot\delta)
&=
\var(\hat\delta)
+\frac2{N}\sum_{j<k}
\biggl(\,
\frac{\eta_j\eta_k}4\prod_{\ell \in(j,k)}\mu_{2\ell }^2\prod_{\ell \not\in[j,k]}\mu_{4\ell }
-\sigma^2_j\sigma^2_k\mu_{2j}\mu_{2k}
\prod_{\ell \not\in j{:}k}\mu_{2\ell }^2
\biggr)
\end{align}
and
\begin{align}\label{eq:varmultwindnummain}
\var( \check\delta)
&=\var( \ddot\delta)
+\frac2{N}\sum_{j<k}
\biggl(\,
\frac{\eta_j\eta_k}4\prod_{\ell \not\in j{:}k}\mu_{4\ell }
-\sigma^2_j\sigma^2_k
\prod_{\ell \not\in j{:}k}\mu_{2\ell }^2
\biggr)
\prod_{\ell \in(j,k)}\mu^2_{2\ell }.
\end{align}
\end{theorem}
\begin{proof}
The winding stairs results are from Theorem~\ref{thm:multiplicativecasewinding}
in the Appendix.
Next we turn to the naive estimator.
For $\bsx,\bsz\sim P$ independently,
define
$\Delta_j=\Delta_j(\bsx,\bsz)
\equiv f_P(\bsx)-f_P(\bsx_{-j}{:}\bsz_j)$.
Now
\begin{align*}
\Delta_j & = (g_j(x_j)-g_j(z_j))
\times\prod_{\ell\ne j}g_\ell(x_\ell)
\end{align*}
and so
$\e(\Delta_j^2)=2\sigma^2_j\times\prod_{\ell\ne j}\mu_{2\ell}$
and $\e(\Delta_j^4)=(12+2\kappa_j)\sigma^4_j\times
\prod_{\ell\ne j}\mu_{4j}$, from Lemma~\ref{lem:mu4s}.
Therefore
\begin{align*}
\var(\Delta_j^2) =
(12+2\kappa_j)\sigma^4_j\times \prod_{\ell\ne j}\mu_{4j}
-
4\sigma^4_j\times\prod_{\ell\ne j}\mu_{2\ell}^2.
\end{align*}
establishing~\eqref{eq:varnaivprod}.

In the radial estimate, $\Delta_j$ is as above and
$\Delta_k = (g_k(x_k)-g_k(z_k))\times\prod_{\ell\ne k}g_{\ell}(x_{\ell})$.
In this case however the same point $\bsx$ is used in both
$\Delta_j$ and $\Delta_k$ so
$\e(\Delta_j^2\Delta_k^2)$ equals
\begin{align*}
&
\e\Bigl(g_j(x_j)^2g_k(x_k)^2(g_j(x_j)-g_j(z_j))^2(g_k(x_k)-g_k(z_k))^2
\prod_{\ell\not\in\{j,k\}} g_\ell(x_\ell)^4
\Bigr)\\
&=
\eta_j\eta_k\prod_{\ell\not\in\{j,k\}}\mu_{4\ell}.
\end{align*}
Then
$\cov(\Delta_j^2,\Delta_k^2)  =
\bigl(\eta_j\eta_k
-4\sigma_j^2\sigma_k^2\mu_{2j}\mu_{2k}\bigr)
\prod_{\ell\not\in\{j,k\}}\mu_{4\ell}$,
establishing~\eqref{eq:varradiprod}.
\end{proof}

We comment below on interpretations of the winding
stairs quantities. First we compare naive to radial
sampling.

As an illustration, suppose that $g_j(x_j)\sim\dnorm(0,1)$
for $j=1,\dots,d$.
Then
\begin{align*}
\var( \hat\delta) &=
\frac1N\sum_{j=1}^d(3^d-1)=\frac{d(3^d-1)}N
\end{align*}
and since this example has $\eta_j = 4$,
\begin{align*}
\var( \tilde\delta) &=
\frac{d(3^d-1)}N
+\frac2N\sum_{j<k}\Bigl( \frac{16}4-1\Bigr) 3^{d-2}
%\\&
=
\frac{d(3^d-1)}N
+\frac{2d(d-1)3^{d-1}}N.
\end{align*}
For large $d$ the radial method has variance about $2d/3$ times
as large as the naive method. Accounting for the reduced sample
size of the radial method it has efficiency approximately $3/d$
compared to the naive method, for this function.

A product of mean zero functions has mean dimension $d$
making it an exceptionally hard case.
More generally, if $\eta_j/2-\sigma^2_j\mu_{2j}\ge\epsilon >0$
for $j\in1{:}d$, then $\var(\hat\delta)=O(d/N)$ while
$\var(\tilde\delta)$ is larger than a multiple of $d^2/N$.

\begin{corollary}
For the product function $f_P$ of~\eqref{eq:fprod},
suppose that $\kappa_j\ge -5/16$ for $j=1,\dots,d$.
Then $\cov(\wt\olt^2_j, \wt\olt^2_k)\ge0$  % to go for >0 we need to assume more
for $1\le j<k\le d$,
and so $\var( \tilde\delta)\ge\var(\hat\delta)$.
\end{corollary}
\begin{proof}
It suffices to show that
$\eta_j>2\sigma^2_j\mu_{2j}$ for $j=1,\dots,d$.
Let $Y=g_j(x_j)$ for $x_j\sim P_j$
have mean $\mu$, uncentered moments
 $\mu_{2y}$, $\mu_{3y}$ and $\mu_{4y}$
of orders $2$, $3$ and $4$, respectively,
variance $\sigma^2$, skewness $\gamma$, and kurtosis $\kappa$.
Now let $\eta = \mu_{4y}-2\mu\mu_{3y}+\mu_{2y}^2$.
This simplifies to
%\begin{align*}
$$\eta
%& = \bigl((\kappa+3)\sigma^4+4\mu\gamma\sigma^3+6\mu^2\sigma^2+\mu^4\bigr)
% -2\mu\bigl(\gamma\sigma^3+3\mu\sigma^2+\mu^3\bigr)
%+\bigl( \mu^4+2\mu^2\sigma^2+\sigma^4\bigr)\\&
=
(\kappa+2)\sigma^4+2\mu\sigma^3\gamma +2\mu^2\sigma^2+\sigma^4$$
%\end{align*}
and so
$$
\eta-2\sigma^2\mu_{2y}=
(\kappa+2)\sigma^4+2\mu\sigma^3\gamma +\mu^2\sigma^2.
$$

If $\sigma=0$ then $\eta-2\sigma^2\mu_{2y}=0$ and so
we suppose that $\sigma>0$.
Replacing $Y$ by $Y/\sigma$ does not change the sign of $\eta-2\sigma^2\mu_{2y}$.
It becomes
$\kappa+2+2\mu_*\gamma +\mu_*^4$ for $\mu_*=\mu/\sigma$.
If $\gamma$ and $\mu_*$ have equal signs, then
$\kappa+2+2\mu_*\gamma +\mu_*^4\ge0$, so we consider
the case where they have opposite signs. Without loss of generality
we take $\gamma < 0 <\mu_*$.
An inequality of \cite{roha:szek:1989}
shows that $|\gamma|\le \sqrt{\kappa+2}$ and so
\begin{align}\label{eq:postrohaszek}
\kappa+2+2\mu_*\gamma +\mu_*^4
&\ge\theta^2-2\mu_*\theta +\mu_*^4
\end{align}
for $\theta=\sqrt{\kappa+2}$.
Equation~\eqref{eq:postrohaszek} is minimized over $\mu_*\ge0$
at $\mu_* = (\theta/2)^{1/3}$ and so
$\kappa+2+2\mu_*\gamma +\mu_*^4
%&\ge\theta^2-2\theta\Bigl(\frac\theta2\Bigr)^{1/3} +\Bigl(\frac\theta2\Bigr)^{4/3}
\ge\theta^2
+\bigl(2^{-4/3}-2^{2/3}\bigr)\theta^{4/3}$.
One last variable change to $\theta = (2\lambda)^3$
gives
$$\kappa+2+2\mu_*\gamma +\mu_*^4
\ge \lambda^4(4\lambda^2-3).$$
This is nonnegative for $\lambda\ge(3/4)^{1/2}$,
equivalently $\theta\ge 2(3/4)^{3/2}$ and finally
for $\kappa\ge -5/16$.
\end{proof}

From the above discussion we can see that large kurtoses
and hence large values of $\mu_{4j}=\e( g_j(x_j)^4)$ create difficulties.
In this light we can compare winding stairs to the
radial sampler.  The covariances in the radial sampler
involve a product of $d-2$ of the $\mu_{4j}$.
The winding stairs estimates involve products
of fewer of those quantities. For truncated winding
stairs the $j,k$-covariance include a product of only
$d-k+j-1$ of them. The values $\mu_{4\ell}$ for $\ell$
nearest to $1$ and $d$ appear the most often and so
the ordering of the variables makes a difference.
For regular winding stairs some additional fourth
moments appear in a second term.

%\section{Quadratic functions}\label{sec:quadratic}

%Here we consider quadratic functions
%$f_Q(\bsx) = \bsx^\tran A\bsx$ where $A\in\real^{d\times d}$
%is a symmetric matrix.  Fourth moments of $f_Q$ will involve
%chaining. We consider the case with $\bsx\sim\dnorm(0,I)$.

%{\bf Chris has some results for here}

\section{Example: MNIST classification}\label{sec:neural}

In this section, we investigate the mean dimension
of a neural network classifier that predicts a
digit in $\{0,1,\dots,9\}$ based on an image
of $784$ pixels.
We compare algorithms for finding mean dimension,
investigate some mean dimensions, and then plot
some images of Sobol' indices.

The MNIST data set from
\url{http://yann.lecun.com/exdb/mnist/}
is a very standard benchmark problem for neural
networks.  It consists
of 70,000 images of hand written digits that
were size-normalized and centered within
$28\times 28$ pixel gray scale images.
We normalize the image values to the unit interval, $[0,1]$.
The prediction problem is to identify which
of the ten digits `0', `1', $\dots$, '9' is in one of
the images based on $28^2=784$ pixel values.
We are interested in the mean dimension of a fitted
prediction model.

The model we used is a convolutional neural network
fit via tensorflow
\citep{abadi2016tensorflow}.
The architecture applied the following steps
to the input pixels in order:
\begin{compactenum}[\quad\bf1)]
\item
   a convolutional layer (with 28 kernels, each of size 3x3),
\item
   a max pooling layer (over 2x2 blocks),
\item
   a flattening layer,
\item
   a fully connected layer with 128 output neurons (ReLU activation),
\item
   a dropout layer (node values were set to 0 with probability 0.2), and
\item
   a final fully connected layer with 10 output neurons (softmax activation).
\end{compactenum}
This model is from~\cite{yalcin2018image} who also defines those terms.
The network was trained
using 10 epochs of ADAM optimization, also described in~\cite{yalcin2018image},
on 60,000 training images.
For our purposes, it is enough to know that it is a complicated
black box function of $784$ inputs.
The accuracy on 10,000 held out images was 98.5\%.  This is not necessarily
the best accuracy attained for this problem, but we consider it good enough
to make the prediction function worth investigating.

There are $2^{784}-1>10^{236}$ nontrivial sets of pixels, each making
their own contribution to the prediction functions, but the mean
dimension can be estimated by summing only $784$ Sobol' indices.

We view the neural network's prediction as a function on $784$
input variables $\bsx$.
For data $(\bsx,Y)$ where
$Y\in\{0,1,\dots,9\}$ is the true digit of the image,
the estimated probability
that $Y=y$ is given by
$$
f_y(\bsx) =
\frac{\exp( g_y(\bsx)) }{\sum_{\ell=0}^9\exp(g_\ell(\bsx))}.
$$
for functions $g_y$, $0\le y\le 9$.
This last step, called the softmax layer, exponentiates
and normalizes functions $g_y$ that implement the prior layers.
We study the mean dimension of $g_0,\dots,g_9$
as well as the mean dimensions of $f_0,\dots,f_9$.
Studying the complexity of predictions via the inputs
to softmax has been done earlier by
\cite{yosi:etal:2015}.

To compute mean dimension we need to
have a model for $\bsx$ with $784$ independent
components. Real images are only on or near a
very small manifold within $\real^{784}$.
We considered several distributions $P_j$
for the value of pixel $j$:
$\dustd\{0,1\}$ (salt and pepper)
$\dustd[0,1]$ (random gray),
independent resampling from per pixel
histograms of all images,
and independent resampling per pixel just
from images with a given value of $y\in\{0,1,\dots,9\}$.
The histogram of values for pixel $j$ from those images
is denoted by $h_y(j)$ with $h_y$ representing all $784$ of them.
Figure~\ref{fig:sampleexamples} shows some
sample draws along with one real image.
We think that resampling pixels from images given $y$
is the most relevant of these methods, though ways
to get around the independence assumption would be valuable.
We nonetheless include the other samplers in our computations.

\begin{figure}
\centering
\hspace*{-.35in}
\includegraphics[width=1.1\hsize]{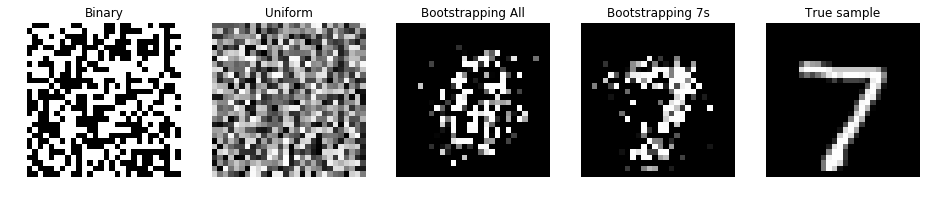}
\caption{\label{fig:sampleexamples}
From left to right: draws from
$\dustd\{0,1\}^{28\times28}$,
$\dustd[0,1]^{28\times28}$,
margins of all images,
margins of all $7$s,
an example $7$.
}
\end{figure}

Our main interest is in comparing the variance of estimates
of $\delta$. We compared the naive method $\hat\delta$,
the radial method $\tilde\delta$ and truncated winding
stairs $\ddot\delta$.
For $\ddot\delta$ our winding stairs algorithm
changed pixels in raster order, left to right within
rows, taking rows of the image from top to bottom.
We omit $\check\delta$ because we think there is
no benefit from its more complicated model and additional correlations.
Our variance comparisons are based on $N=100{,}000$
samples.

\begin{figure}[t]
\centering
\includegraphics[width=.9\hsize]{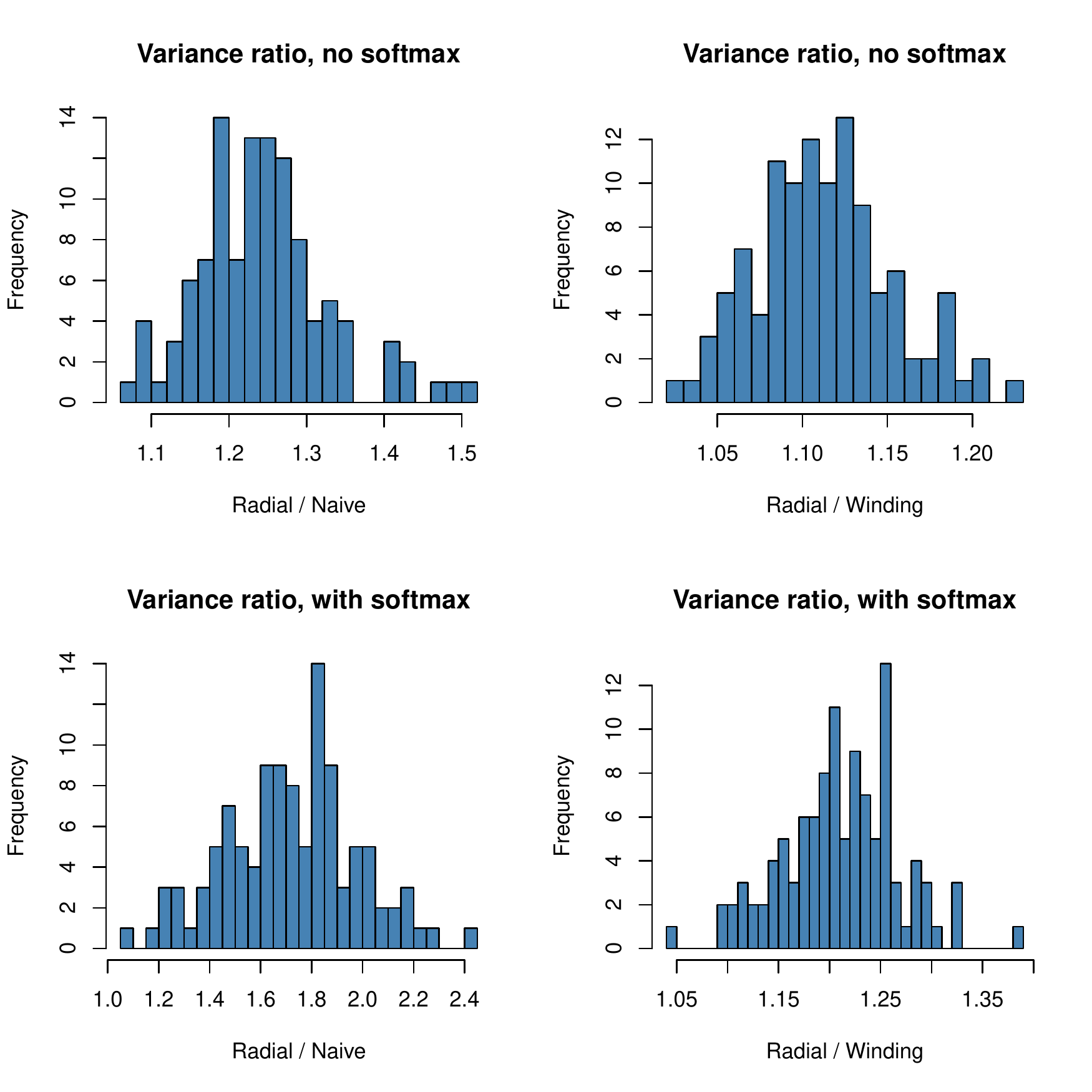}
\caption{\label{fig:varall}
The upper left histogram shows $\var(\tilde\delta)/\var(\hat\delta)$
for functions $g_y$ that exclude softmax.
The upper right histogram shows $\var(\tilde\delta)/\var(\ddot\delta)$.
The bottom two show the same ratios for functions $f_y$ that
include softmax. The histograms include all $10$ values of output $y$,
and all $10$ $y$-specific input histograms and the pooled
input histogram.
}
\end{figure}

Figure~\ref{fig:varall} shows the results for all $10$ output values
$y$, and all $11$ different input histogram distributions.
Ten of those histograms are from resampling pixel values within
categories and the eleventh is a pooled histogram.
There are separate plots
for functions $f_y$ that include softmax and $g_y$ that exclude it.
The radial method always had greater variance than the naive
method. For functions $g_y$ it never had as much as twice the
variance of the naive method, and so the radial method proves
better for $g_y$.  For $f_y$ there were some exceptions where
the naive method is more efficient.
In all of our comparisons the winding stairs method had
lower variance than the radial method, and so for these
functions, (truncated) winding stairs is clearly the best choice.

Figure~\ref{fig:varall} is a summary of $660$ different variance
estimates.  We inspected the variances and found two more
things worth mentioning but not presenting.
The variances were all far smaller using softmax than not, which is not surprising
since softmax compresses the range of $f_y$ to be within $[0,1]$
which will greatly affect the differences that go into estimates of $\delta$.
The variances did not greatly depend on the input distribution.
While there were some statistically significant differences,
which is almost inevitable for such large $N$,
the main practical difference was that variances tended to be much
smaller when sampling from $h_1$.  We believe that this is because images for $y=1$ have much less
total illumination than the others.

While our main purpose is to compare estimation strategies
for mean dimension, the mean dimensions for this problem
are themselves of interest.
Table~\ref{tab:withsoftmax} shows mean dimensions for
functions $f_y$ that include softmax
as estimated via winding stairs.  For this we used $N=10^6$ when
resampling from images $h_0,\dots,h_9$ and $N=2\times 10^6$ otherwise.
The first thing to note is an impossible estimate of $\nu(f_1)$
for binary and uniform sampling. The true $\nu(f_1)$ cannot be larger than $784$.
The function $f_1$ has tiny variance under those distributions and recall
that $\nu=\delta/\sigma^2$.
Next we see that moving from binary to uniform to the combined histogram
generally lowers the mean dimension.
Third, for the $y$-specific histograms $h_y$ we typically see smaller mean
dimensions for $f_y$ with the same $y$ that was used in sampling.
That is, the diagonal of the lower block tends to have smaller values.

\begin{table}
%\begin{tabular}{r|cccccccccc}
  \begin{tabular}{r|rrrrrrrrrrrrrr}
Sampler         & 0             & 1                     & 2             & 3             & 4             & 5             & 6
        & 7             & 8             & 9     \\
\hline
binary          & 11.07 & 936.04        & 10.43 & 9.92  & 18.69 & 10.22 & 13.27 & 13.37 & 8.67  & 16.54 \\
uniform         & 6.92  & 4,108.99      & 7.28  & 6.60  & 9.90  & 7.03  & 6.92  & 8.03  & 5.61  & 9.48  \\
combined        & 8.77  & 4.68          & 4.06  & 3.95  & 4.56  & 5.11  & 7.62  & 4.62  & 3.43  & 7.39  \\
\hline
0                       & 3.52  & 6.81          & 3.48  & 7.20  & 6.56  & 5.78  & 7.54  & 4.67  & 4.04  & 9.08  \\
1                       & 36.12 & 2.88          & 6.00  & 3.43  & 7.75  & 3.76  & 8.74  & 7.60  & 2.83  & 5.58  \\
2                       & 10.03 & 3.86          & 3.68  & 4.70  & 8.23  & 12.27 & 12.57 & 7.20  & 4.31  & 17.23 \\
3                       & 23.20 & 4.69          & 5.95  & 4.10  & 6.96  & 6.72  & 13.63 & 7.10  & 4.42  & 9.00  \\
4                       & 7.42  & 8.39          & 7.59  & 9.96  & 3.81  & 7.63  & 8.57  & 5.35  & 3.86  & 6.82  \\
5                       & 8.12  & 4.77          & 5.72  & 4.82  & 5.60  & 3.48  & 7.61  & 7.28  & 3.54  & 7.87  \\
6                       & 9.22  & 5.65          & 4.36  & 6.52  & 4.31  & 6.67  & 3.57  & 6.43  & 4.28  & 11.99 \\
7                       & 8.57  & 5.85          & 4.42  & 4.09  & 4.66  & 5.09  & 3.59  & 3.59  & 4.29  & 5.58  \\
8                       & 19.58 & 6.06          & 4.54  & 4.77  & 8.21  & 6.28  & 13.15 & 6.72  & 4.20  & 10.11 \\
9                       & 7.47  & 7.00          & 5.25  & 4.96  & 3.15  & 4.52  & 7.34  & 3.74  & 2.92  & 3.48\\
\hline
\end{tabular}
\caption{\label{tab:withsoftmax}
Estimated mean dimension of functions $f_y$ using softmax.
}
\end{table}

Table~\ref{tab:withoutsoftmax} shows mean dimensions for
functions $g_y$ that exclude softmax
as estimated via winding stairs.
They are all in the range from $1.35$ to $1.92$.
We found no particular problem with the function $g_1$ like we saw for $f_1$.
While the functions $g_y$ that are sent into softmax were obtained
by a very complicated process, they do not make much use of very
high order interactions.  There must be a significantly large component
of additive functions and two factor interactions within them. There may be a small
number of large high order interactions but they do not dominate
any of the functions $f_y$ under any of the sampling distributions we use.
The softmax function begins by exponentiating $f_y$ which we can think
of as changing a function with a lot of additive structure into one with
a lot of multiplicative structure.  Multiplicative functions can have quite
high mean dimension.

The measured mean dimensions of $g_y$ are pretty stable as the sampling
distribution changes.  While the manifold of relevant images is likely to
be quite small, it is reassuring that $13$ different independent data distributions
give largely consistent and small mean dimensions.

\begin{table}
\begin{tabular}{r|cccccccccc}
Sampler         & 0             & 1                     & 2             & 3             & 4             & 5             & 6
        & 7             & 8             & 9     \\
\hline
binary          & 1.66  & 1.76  & 1.74  & 1.72  & 1.73  & 1.79  & 1.75  & 1.69  & 1.74  & 1.79 \\
uniform         & 1.65  & 1.62  & 1.66  & 1.66  & 1.67  & 1.71  & 1.71  & 1.61  & 1.68  & 1.70 \\
combined        & 1.79  & 1.77  & 1.70  & 1.73  & 1.73  & 1.90  & 1.88  & 1.78  & 1.90  & 1.89 \\
\hline
0                       & 1.92  & 1.65  & 1.68  & 1.69  & 1.65  & 1.80  & 1.86  & 1.56  & 1.68  & 1.81 \\
1                       & 1.48  & 1.56  & 1.35  & 1.61  & 1.62  & 1.57  & 1.49  & 1.42  & 1.56  & 1.50 \\
2                       & 1.55  & 1.66  & 1.62  & 1.74  & 1.57  & 1.72  & 1.67  & 1.61  & 1.78  & 1.59 \\
3                       & 1.56  & 1.65  & 1.59  & 1.58  & 1.63  & 1.85  & 1.59  & 1.64  & 1.67  & 1.66 \\
4                       & 1.87  & 1.62  & 1.61  & 1.55  & 1.70  & 1.75  & 1.76  & 1.66  & 1.57  & 1.78 \\
5                       & 1.71  & 1.60  & 1.59  & 1.63  & 1.72  & 1.78  & 1.74  & 1.62  & 1.76  & 1.90 \\
6                       & 1.65  & 1.60  & 1.60  & 1.66  & 1.68  & 1.70  & 1.65  & 1.60  & 1.54  & 1.63 \\
7                       & 1.73  & 1.59  & 1.61  & 1.63  & 1.60  & 1.62  & 1.65  & 1.57  & 1.59  & 1.63 \\
8                       & 1.73  & 1.65  & 1.60  & 1.64  & 1.66  & 1.78  & 1.75  & 1.64  & 1.84  & 1.75 \\
9                       & 1.86  & 1.68  & 1.61  & 1.63  & 1.73  & 1.80  & 1.86  & 1.67  & 1.69  & 1.82\\
\hline
\end{tabular}
\caption{\label{tab:withoutsoftmax}
Estimated mean dimension of functions $g_y$ without softmax.
}
\end{table}

Figure~\ref{fig:sometau} shows some Sobol' indices
of $f_y$ and $g_y$ for $y\in\{0,1,\dots,9\}$ when sampling
from $h_0$.  In each set of $10$ images, the gray scale goes
from black for $0$ to white for the largest intensity in any of those $10$
images.  As a consequence some of the images are almost entirely black.

The lower indices $\ult^2_j$ depict the importance
of inputs one at a time.  This is similar to what one gets
from a gradient,
see for instance  Grad-cam \citep{selvaraju2017grad},
except that $\ult^2_j$ is global over the whole range of the
input instead of local like a gradient.
Upper indices $\olt^2_j$ depict the importance
of each pixel combining all of the interactions
to which it contributes, not just its main effect.

\begin{figure}[t!]
\centering
\includegraphics[width=.9\hsize]{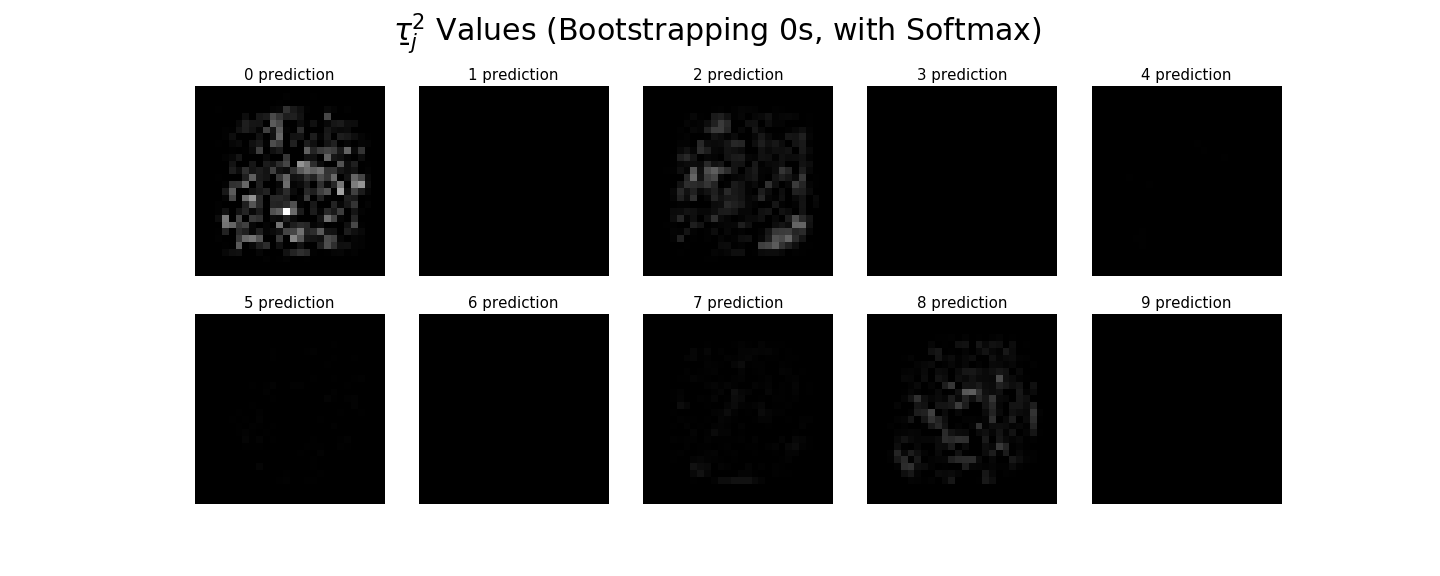}\\[-1ex]
\includegraphics[width=.9\hsize]{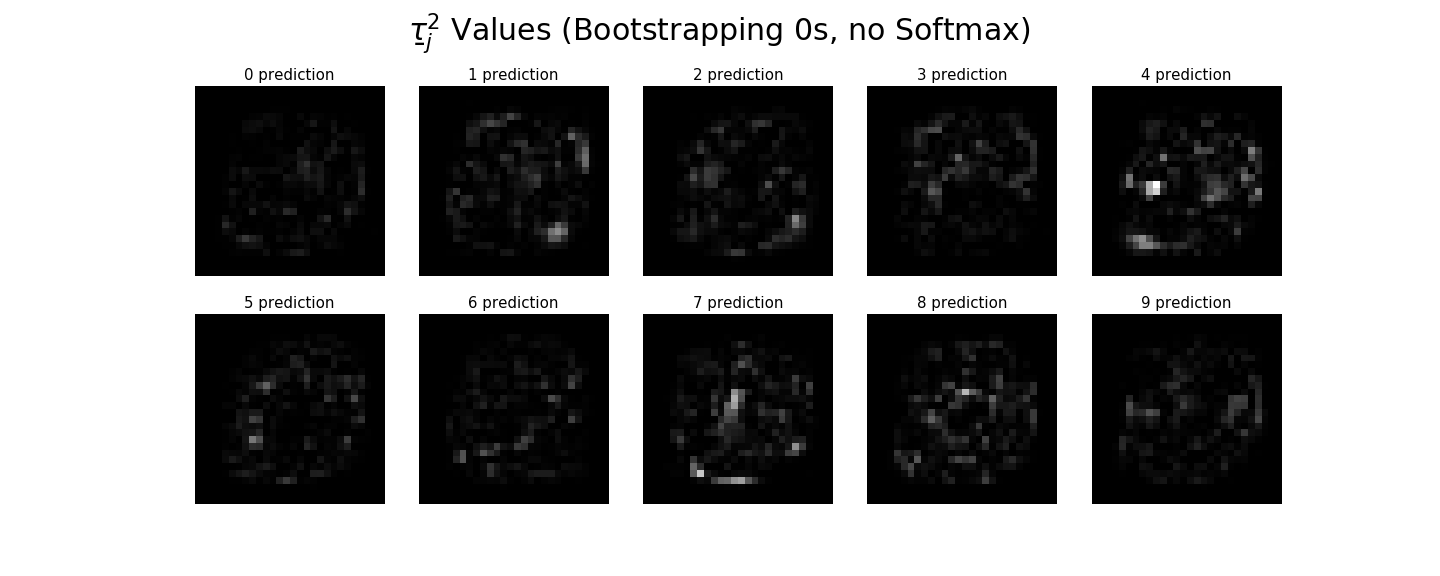}\\[-1ex]
\includegraphics[width=.9\hsize]{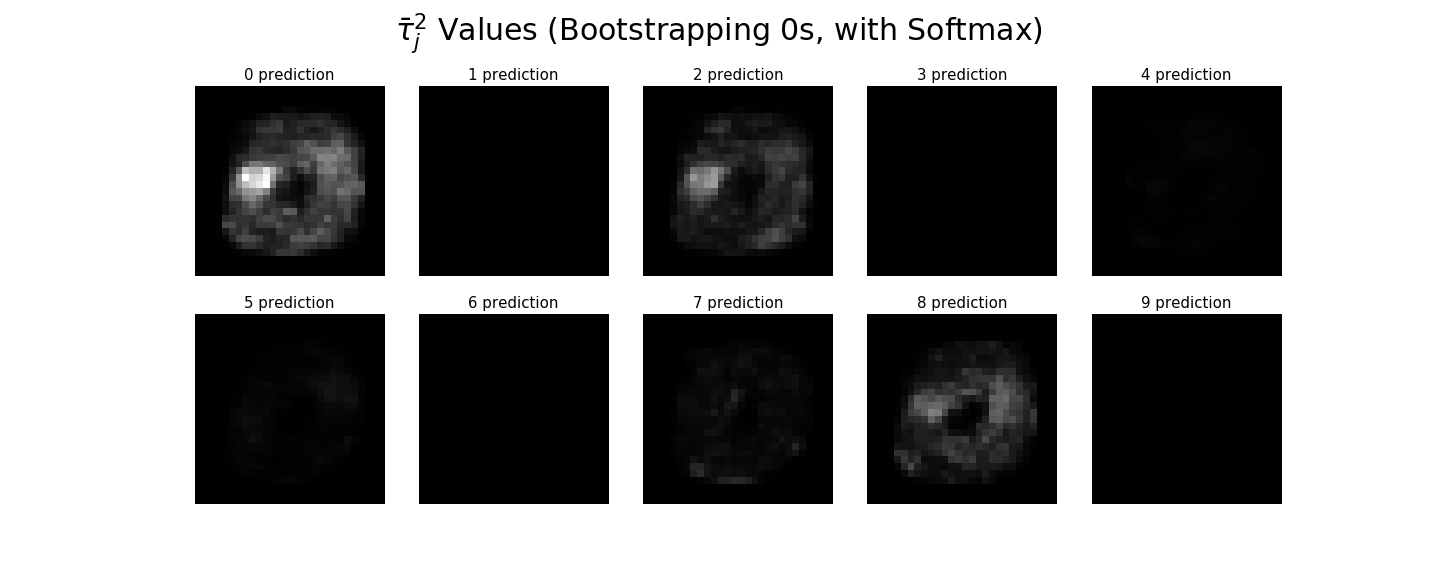}\\[-1ex]
\includegraphics[width=.9\hsize]{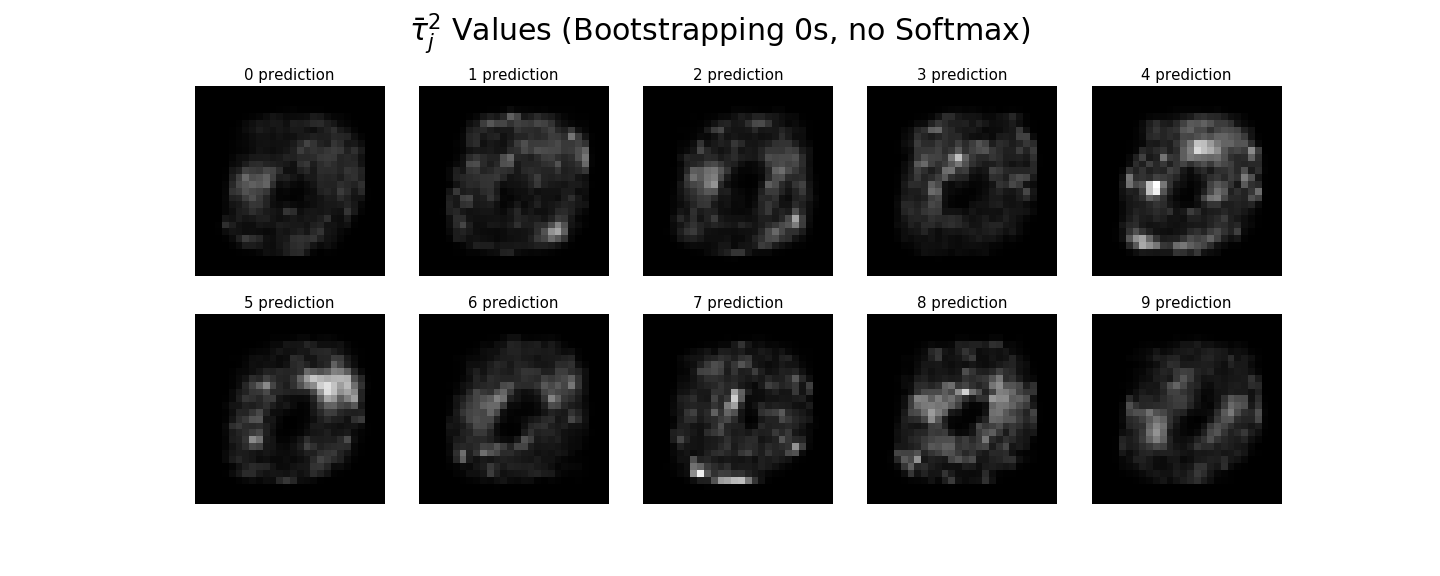}
\caption{\label{fig:sometau}
From top to bottom:
maps of $\ult^2_j(f_y)$, $\ult^2_j(g_y)$,
$\olt^2_j(f_y)$ and $\olt^2_j(g_y)$ versus pixels $j$
when sampling from $h_0$.
}
\end{figure}

For the influence on $f_0$ when sampling from $h_0$, the difference between
$\ult^2_j$ and $\olt^2_j$ is in that bright spot just left of the center of
the image. That is the region of pixels involved in the most interactions.
It appears to be involved in distinguishing $0$s from $2$s and $8$s
because that region is also bright for functions $f_2$ and $f_8$.
Without softmax that bright spot for $\olt^2_j$  is lessened and so we
see that much though not all of its interaction importance was introduced
by the softmax layer. For $g_5$ when sampling from $h_0$ we see that a region just Northeast
of the center of the image has the most involvement in interactions
as measured by $\olt^2_j$.

\section{Discussion}\label{sec:discussion}

We have found that the strategy under which differences of function values
are collected can make a big difference to the statistical efficiency
of estimates of mean dimension.  Computational efficiency in reusing
function values can increase some correlations enough to more
than offset that advantage. Whether this happens depends on the function
involved. We have seen examples where high kurtoses make the problem worse.

Our interest in mean dimension leads us to consider sums of $\olt^2_j$.
In other uncertainty quantification problems we are interested in comparing
and ranking $\olt^2_j$.  For a quantity like $\hat\olt^2_j-\hat\olt^2_k$
we actually prefer a large positive value for $\cov(\hat\olt^2_j,\hat\olt^2_k)$.
In this case, the disadvantages we described for the radial method
become a strength.  Correlation effects are more critical
for mean dimension than for these differences of Sobol'
indices, because mean dimension is affected by $O(d^2)$ covariances, not just one.

The radial strategy and the truncated winding stairs strategy can both be
represented in terms of a tree structure connecting $d+1$ function values.
There is a one to one correspondence between the $d$ edges in that tree and
the components of $\bsx$ getting changed.
There is no particular reason to think that either of these strategies is
the optimal graph structure or even the optimal tree.

The mean dimension derives from an ANOVA decomposition that in turn is based
on models with independent inputs.  There has been work on ANOVA for
dependent inputs, such as \cite{ston:1994}, \cite{hook:2007}
and \cite{chas:gamb:prie:2012,chas:gamb:prie:2015}. The underlying
models require the density to have an unrealistically strong absolute continuity
property with respect to a product measure that makes them unrealistic
for the MNIST example.
There are also approaches to global sensitivity analysis
based on Shapley values, that do not require independence
of the underlying variables
\citep{song:nels:stau:2016,shapleydependent}.

Recent work by \cite{hart:grem:2018} shows how to define some Sobol' indices
directly without recourse to the ANOVA and that may provide a basis for
mean dimension without ANOVA.
\cite{kuch:tara:anno:2012} have a copula based approach to Sobol' indices
on dependent data, though
finding a specific copula that describes points near a manifold would be hard.

We have studied the accuracy of mean dimension estimates as if
the sampling were done by plain Monte Carlo (MC).
When $P$ is the uniform distribution on $[0,1]^d$
then we can instead use randomized quasi-Monte Carlo (RQMC)
sampling, surveyed in \cite{lecu:lemi:2002}.  The naive method can be implemented using $N$ points in $[0,1]^{d+1}$ for each of $j=1,\dots,d$.
The first column of the $j$'th input matrix could
contain $\bsz_{ij}$ for $i=1,\dots,N$ while the remaining $d$
columns would have $\bsx_i^{(j)}\in[0,1]^d$.
The $d+1$'st point contains the values $\bsx_{i,j}$.
The radial method can be implemented with $N$ points in $[0,1]^{2d}$
with the first $d$ columns providing $\bsx_i$ and the second $d$
columns providing $\bsz_i$, both for $i=1,\dots,N$.
Truncated winding stairs, similarly requires $N$ points in $[0,1]^{2d}$.
For RQMC sampling by scrambled nets, the resulting variance is $o(1/N)$.
A reasonable choice is to use RQMC in whichever method one thinks
would have the smallest MC variance.  The rank ordering
of RQMC variances could however be different from that of MC and
it could even change with $N$, so results on MC provide only
a suggestion of which method would be best for RQMC.

A QMC approach to
plain winding stairs would require QMC methods
designed specifically for MCMC sampling.   See for instance,  one
based on completely uniformly distributed
sequences described in \cite{qmc4mcmc}.

We have used a neural network black box function to illustrate our computations.
It is yet another example of an extremely complicated function that nonetheless
is dominated by low order interactions.
In problems like this where the input images had a common registration
an individual pixel has some persistent meaning between images
and then visualizations of $\ult^2_j$ can be informative.
Many neural network problems are applied to data that have not
been so carefully registered as the MNIST data.  For those problems
the link from predictions back to inputs may need to be explored
in a different way.

\section*{Acknowledgments}
This work was supported by
a grant from Hitachi Limited and by
the US National Science Foundation
under grant IIS-1837931.
We thank Masayoshi Mase of Hitachi for helpful discussions about
variable importance and explainable AI.
We also thanks anonymous reviewers for suggestions that have
improved our presentation.
\bibliographystyle{apalike}
\bibliography{sensitivity}

\section*{Appendix: Covariances under winding stairs}

Winding stairs expressions are more complicated than
the others and require somewhat different notation.
Hence we employ some notation local to this appendix.
For instance in winding stairs $\ell(i)$ has a special
meaning as newly updated component of $\bsx_i$.
Accordingly when we need a variable index other than $j$
and $k$ we use $t$ instead of $\ell$, in this appendix.
We revert the $t$'s back to $\ell$ when quoting
these theorems in the main body of the paper.
Similarly, differences in function values are more
conveniently described via which observation $i$ is involved
and not which variable.  Accordingly, we work with $\Delta_i$
here instead of $\Delta_j$ in the main body of the article.

We begin with the regular winding
stairs estimates and let
$\Delta_i = f(\bsx_i)-f(\bsx_{i-1})$.
For $i'>i$, the differences $\Delta_i$ and $\Delta_{i'}$
are independent if $\bsx_{i'-1}$ has no common
components with $\bsx_i$.
This happens if $i'-1\ge i+d$, that is if $i'-i>d$.
For any index $i$, the difference $\Delta_i$
may be dependent on $\Delta_{i'}$ for $-d<i'<d$
but no other $\Delta_{i'}$.
It is not necessarily true that
$\cov(\Delta_i^2,\Delta_{i+s}^2)=\cov(\Delta_i^2,\Delta_{i-s}^2)$
because different shared components of $\bsx$ are involved
in these two covariances.

The winding stairs estimate of $\olt^2_j$ is
$
{\check\olt}^2_j = (1/(2N))\sum_{i=1}^N\Delta_{d(i-1)+j}^2.
$
Because $\cov(\Delta_{i+d}^2,\Delta_{i'+d}^2)
=\cov(\Delta_{i}^2,\Delta_{i'}^2)$,
we find that for $1\le j<k\le d$,
\begin{align}\label{eq:covoltwinding}
\cov( \check{\olt}^2_j, \check{\olt}^2_k)
&=
\frac1{4N}\Bigl(
\cov(\Delta_{d+j}^2 ,\Delta_{d+k}^2)
+\cov(\Delta_{2d+j}^2 ,\Delta_{d+k}^2)\Bigr).
\end{align}
The truncated winding stairs algorithm has
\begin{align}\label{eq:covoltwinding}
\cov( \ddot{\olt}^2_j, \ddot{\olt}^2_k)
&=
\frac1{4N}
\cov(\Delta_{d+j}^2 ,\Delta_{d+k}^2)
\end{align}
because $\Delta_{2d+j}$ has no $z$'s in
common with $\Delta_{d+k}$.

\begin{theorem}\label{thm:additivecasewinding}
For the additive function $f_A$ of~\eqref{eq:fadd},
\begin{align}
\var( \check\delta)
&=\frac1N\sum_{j=1}^d\Bigl(2+\frac{\kappa_j}2\Bigr)\sigma^4_j
+\frac{N-1}{2N^2}\sum_{j=1}^d(\kappa_j+2)\sigma^4_j
\label{eq:varaddwindnum}\\
\var( \ddot\delta)
&=\frac1N\sum_{j=1}^d\Bigl(2+\frac{\kappa_j}2\Bigr)\sigma^4_j.\label{eq:varadddiswindnum}
\end{align}
\end{theorem}
\begin{proof}
For an additive function under winding stairs
\begin{align*}
\Delta_{d(i-1)+j}
&=g_{j}(\bsx_{d(i-1)+j,j})-g_{j}(\bsx_{d(i-2)+j,j})\\
&=g_{j}(z_{d(i-1)+j})-g_{j}(z_{d(i-2)+j})
\end{align*}
because $r(i,j) = d\lfloor (i-j)/d\rfloor+j$
yields $r(d(i-1)+j,j)=d(i-1)+j$.
It follows that $\check{\olt}^2_j$
and $\check{\olt}^2_k$ have no $z$'s in common
when $j\ne k$ and so they are independent.
Now define the independent and identically
distributed random variables
$Y_i=g_j(z_{d(i-1)+j})$ for $i=1,\dots,N$.
Then
\begin{align*}
\var(\check{\olt}^2_j)
&=\var\Bigl(\frac1{2N}\sum_{i=1}^N(Y_i-Y_{i-1})^2\Bigr)\\
&=\frac1{4N}\var( (Y_1-Y_0)^2)
+\frac{N-1}{2N^2}
\cov( (Y_1-Y_0)^2, (Y_2-Y_1)^2)\\
&= \frac{(8+2\kappa_j)\sigma^4}{4N}
+\frac{(N-1)(\kappa+2)\sigma^4}{2N^2}
\end{align*}
by Lemma~\ref{lem:mu4s},
establishing~\eqref{eq:varaddwindnum}.
For truncated winding
squares all of the $\Delta_i$ are independent in
the additive model establishing~\eqref{eq:varadddiswindnum}.
\end{proof}

Next we turn to the multiplicative model
$f_P(\bsx_i) = \prod_{j=1}^d g_j( z_{r(i,j)})$.
A key distinction arises for variables `between' the $j$'th
and $k$'th and variables that are not between those.
For $j<k$ the indices $t$ between them are
designated by $t\in(j,k)$ and the ones `outside' of them
are designated by $t\not\in[j,k]$, meaning that
$t\in \{1,\dots,j-1\}\cup\{k+1,\dots,d\}$.
Recall that $\mu_{\ell j}$ is $\e( g_j(x_j)^\ell)$ for $\ell=2,3,4$.

\begin{theorem}\label{thm:multiplicativecasewinding}
For the multiplicative function $f_P$ of~\eqref{eq:fprod},
\begin{align}
\label{eq:varmultdiswindnum}
\begin{split}
\var( \ddot\delta)
&=
\frac1{N}
\sum_{j=1}^d
\sigma_j^4\Bigl(\Bigl(3+\frac{\kappa_j}2\Bigr)\prod_{t\ne j}\mu_{4t}-\prod_{t\ne j}\mu_{2t}^2\Bigr)\Bigr)\\
&+\frac2{N}\sum_{j<k}
\biggl(\,
\frac{\eta_j\eta_k}4\prod_{t\in(j,k)}\mu_{2t}^2\prod_{t\not\in[j,k]}\mu_{4t}
-\sigma^2_j\sigma^2_k\mu_{2j}\mu_{2k}
\prod_{t\not\in\{j,k\}}\mu_{2t}^2
\biggr)
\end{split}
\end{align}
and
\begin{align}\label{eq:varmultwindnum}
\begin{split}
\var( \check\delta)
&=\var( \ddot\delta)
+\frac2{N}\sum_{j<k}
\biggl(\,
\frac{\eta_j\eta_k}4\prod_{t\in (j,k)}\mu_{4t}\prod_{t\not\in [j,k]}\mu_{2t}^2
-\sigma^2_j\sigma^2_k\mu_{2j}\mu_{2k}
\prod_{t\not\in j{:}k}\mu_{2t}^2
\biggr)
\end{split}
\end{align}
where $\eta_j =
\mu_{4j}-2\mu_j\mu_{3j}+\mu_{2j}^2$.
\end{theorem}
\begin{proof}
We use equation~\eqref{eq:covoltwinding} to write covariances
in terms of the first few $\bsx_i$.
For $1\le j\le d$ we have
$\Delta_{d+j}=
\prod_{t=1}^{j-1}g_t(z_{d+t})
\times\bigl( g_j(z_{d+j})-g_j(z_{d})\bigr)
\times\prod_{t=j+1}^dg_t(z_{t})$
so that
\begin{align*}
\e(\Delta_{d+j}^2) &= 2\sigma_j^2\prod_{t\ne j}\mu_{2t}\quad\text{and}\quad
\e(\Delta_{d+j}^4) = (12+2\kappa_j)\sigma^4_j\prod_{t\ne j}\mu_{4t}
\end{align*}
and $\var(\Delta^2_{d+j})=\eta_j\prod_{t\ne j}\mu_{4t} -4\sigma_j^4\prod_{t\ne j}\mu_{2t}^2$.
Then for $1\le j< k\le d$ and using a convention
that empty products are one,
\begin{align*}
\e( \Delta_{d+j}^2\Delta_{d+k}^2)
&=\prod_{t=1}^{j-1}\mu_{4t}
\times\eta_j\times\prod_{t=j+1}^{k-1}\mu_{2t}^2\times\eta_k\times\prod_{t=k+1}^d\mu_{4t}\quad\text{and}\\
\e( \Delta_{2d+j}^2\Delta_{d+k}^2)
&=\prod_{t=1}^{j-1}\mu_{2t}^2
\times\eta_j\times\prod_{t=j+1}^{k-1}\mu_{4t}\times\eta_k\times\prod_{t=k+1}^d\mu_{2t}^2.
\end{align*}
Therefore,
\begin{align*}
\cov(\Delta^2_{d+j},\Delta^2_{d+k})
&=
\eta_j\eta_k\prod_{t\in(j,k)}\mu_{2t}^2
\prod_{t\not\in[j,k]}\mu_{4t}
-4\sigma^2_j\sigma^2_k\mu_{2j}\mu_{2k}\prod_{t\not\in\{j,k\}}\mu_{2t}^2,\quad\text{and}\\
\cov(\Delta^2_{2d+j},\Delta^2_{d+k})
&=
\eta_j\eta_k\prod_{t\in(j,k)}\mu_{4t}\prod_{t\not\in[j,k]}\mu_{2t}^2
\prod_{t=1}^{j-1}\mu_{2t}^2
-4\sigma^2_j\sigma^2_k\mu_{2j}\mu_{2k}\prod_{t\not\in\{j,k\}}\mu_{2t}^2.
\end{align*}
Putting these together establishes the theorem.
\end{proof}

\end{document}